\documentclass[conference,letterpaper]{IEEEtran}

\pdfoutput=1

\usepackage[utf8]{inputenc} 
\usepackage[T1]{fontenc}
\usepackage{url}
\usepackage{ifthen}
\usepackage{cite}
\usepackage[cmex10]{amsmath} %

\usepackage[ruled]{algorithm}
\usepackage{colortbl}
\usepackage{booktabs}
\usepackage{multirow}
\usepackage{caption}
\usepackage{float}
\usepackage{multicol}
\usepackage[normalem]{ulem}
\usepackage{algpseudocode}
\usepackage{tikz-layers}
\algnewcommand{\Initialize}{%
  \State \textbf{Initialize:}
}
\algnewcommand{\Output}{%
  \State \textbf{Output:}
}
\usepackage{mathtools}
\usepackage{graphicx}
\usepackage{url}

\usepackage[
    shortcuts,       %
    acronym,         %
    section=section, %
]{glossaries}
\usepackage{tikz}
\usetikzlibrary{backgrounds,calc,shadings,shapes.arrows,shapes.symbols,shadows,fit,positioning,topaths,hobby}
\usepackage{float}
\usepackage{amssymb}
\usepackage{amsmath}
\usepackage{amsthm}
\usepackage{xspace}
\usepackage[capitalize]{cleveref}
\usepackage{subfig}
\usepackage[inline]{enumitem}
\usepackage{cite}
\usepackage{siunitx}
\usepackage{dsfont}
\usepackage{balance}
\makeglossaries

\newtheorem{theorem}{Theorem}

\newtheorem{lemma}{Lemma}
\newtheorem{observation}{Observation}
\newtheorem{proposition}{Proposition}
\newtheorem{definition}{Definition}

\usepackage{pgfplots}
\pgfplotsset{compat=newest}

\IEEEoverridecommandlockouts

\newcommand{\define}{\ensuremath{:=}}

\DeclareMathOperator*{\argmin}{arg\,min}

\newcommand{\symtx}{\ensuremath{\mathcal{X}}}
\newcommand{\symrx}{\ensuremath{\mathcal{Y}}}

\newcommand{\channel}{\ensuremath{W}}

\newcommand{\card}[1]{\ensuremath{\vert #1 \vert}}

\newcommand{\mutinf}[1]{\mathrm{I}\left(#1\right)}
\newcommand{\entropy}[1]{\mathrm{H}\left(#1\right)}
\newcommand{\kl}[2]{\ensuremath{\mathrm{D}\left(#1 \|#2\right)}}

\newcommand{\chisq}[2]{\ensuremath{\chi^2\left(#1,#2\right)}}

\newcommand{\nearestcenter}[1][x]{\ensuremath{r(#1)}}
\newcommand{\ncenters}{\ensuremath{k}}
\newcommand{\centers}{\ensuremath{\mathcal{R}}}
\newcommand{\pmf}{\ensuremath{W_{\mathrm{Y}|\mathrm{X}}}}
\newcommand{\pmfcondx}[1][x]{\ensuremath{W_{\mathrm{Y}|\mathrm{X}=#1}}}
\newcommand{\ycondx}[1][x]{\ensuremath{W_{\mathrm{Y}|\mathrm{X}=#1}(y)}}

\newcommand{\fullchannel}{\ensuremath{W}}
\newcommand{\prunedchannel}{\ensuremath{W_\centers}}

\newcommand{\E}{\ensuremath{\mathbb{E}}}
\newcommand{\var}{\ensuremath{\text{Var}}}

\newcommand{\extremalpoints}{\ensuremath{\mathcal{E}}}
\newcommand{\redidx}{\ensuremath{u}}
\newcommand{\redsym}{\ensuremath{\mathcal{U}}}
\newcommand{\spanningidx}{\ensuremath{r}}
\newcommand{\spanningsym}{\ensuremath{\mathcal{R}}}
\newcommand{\coeff}{\ensuremath{c_{\redidx, \spanningidx}}}
\newcommand{\allsym}{\ensuremath{\mathcal{X}}}
\newcommand{\allidx}{\ensuremath{x}}

\newcommand{\lbp}[1][x]{\ensuremath{\kappa(#1)}}
\newcommand{\tradeconsta}{\ensuremath{k_1}}
\newcommand{\tradeconstb}{\ensuremath{k_2}}
\newcommand{\pseudocaod}{\ensuremath{Q_{\rvY, \eta}^\star}}
\newcommand{\hulldev}[2]{\ensuremath{\varepsilon_{#1}(#2)}}
\newcommand{\inconvexhullsym}{\ensuremath{\mathcal{I}}}
\newcommand{\inconvexhullidx}{\ensuremath{i}}
\newcommand{\notinconvexhullsym}{\ensuremath{\mathcal{N}}}
\newcommand{\notinconvexhullidx}{\ensuremath{n}}
\newcommand{\closestconditional}[1][\notinconvexhullidx]{W_{Y, #1}^\star}
\newcommand{\closestconditionalvar}[1][\notinconvexhullidx]{W_{Y}}
\DeclareMathOperator{\conv}{Conv}
\newcommand{\hull}[1]{\conv\left(#1 \right)}

\newcommand{\jsd}[2]{\ensuremath{\text{JSD}(#1\Vert #2)}}

\newcommand{\tmprep}[1][x^\star]{\ensuremath{r(#1)}}

\newcommand{\rvZ}{\ensuremath{\mathrm{Z}}}
\newcommand{\tmpsetfull}{\ensuremath{\mathcal{L}}}
\newcommand{\tmpseta}{\ensuremath{\mathcal{J}}}
\newcommand{\tmpsetb}{\ensuremath{\mathcal{K}}}
\newcommand{\tmpsetitem}{\ensuremath{\ell}}
\newcommand{\pruningsym}{\ensuremath{\mathcal{U}}}

\newcommand{\chisqtmp}{\ensuremath{\alpha}}
\newcommand{\nclustersdir}{\ensuremath{\nu}}
\newcommand{\dirscalea}{\ensuremath{d_1}}
\newcommand{\dirscaleb}{\ensuremath{d_2}}
\newcommand{\pmfcondxpruned}[1][x]{\ensuremath{W^\prime_{\mathrm{Y}|\mathrm{X}=#1}}}

\newcommand{\inputdist}{\ensuremath{P_\mathrm{X}}}
\newcommand{\rvY}{\ensuremath{\mathrm{Y}}}
\newcommand{\rvX}{\ensuremath{\mathrm{X}}}
\newcommand{\cluster}{\ensuremath{\mathcal{H}}}
\newcommand{\linkage}[2]{\ensuremath{\text{Link}(#1,#2)}}

\begin{document}
\title{Capacity-Maximizing Input Symbol Selection for Discrete Memoryless Channels}

\author{%
    \IEEEauthorblockN{Maximilian Egger, %
                    Rawad Bitar, %
                    Antonia Wachter-Zeh, %
                    Deniz Gündüz and %
                    Nir Weinberger%
                    } %
    \thanks{M.E., R.B. and A.W-Z. are with the Technical University of Munich. Emails: \{maximilian.egger, rawad.bitar, antonia.wachter-zeh\}@tum.de. D.G. is with Imperial College London. Email: d.gunduz@imperial.ac.uk. N.W. is at Technion --- Israel Institute of Technology. Email: nirwein@technion.ac.il.}
    \thanks{This project has received funding from the German Research Foundation (DFG) under Grant Agreement Nos. BI 2492/1-1 and WA 3907/7-1, and from UKRI for project AI-R (ERC-Consolidator Grant, EP/X030806/1). The work of N.W. was partly supported by the Israel Science Foundation (ISF), grant no. 1782/22.}
}

\maketitle

\begin{abstract}
    Motivated by communication systems with constrained complexity, we consider the problem of input symbol selection for discrete memoryless channels (DMCs). Given a DMC, the goal is to find a subset of its input alphabet, so that the optimal input distribution that is only supported on these symbols maximizes the capacity among all other subsets of the same size (or smaller). We observe that the resulting optimization problem is non-concave and non-submodular, and so generic methods for such cases do not have theoretical guarantees. We derive an analytical upper bound on the capacity loss when selecting a subset of input symbols based only on the properties of the transition matrix of the channel. We propose a selection algorithm that is based on input-symbols clustering, and an appropriate choice of representatives for each cluster, which uses the theoretical bound as a surrogate objective function. We provide numerical experiments to support the findings.%
\end{abstract}
\section{Introduction}

We study the long-standing problem of reducing the input alphabet size of a Discrete Memoryless Channel (DMC) with input alphabet $\mathcal{X}$ to a set of $k < |\mathcal{X}|$ symbols (and possibly $k \ll |\mathcal{X}|$),  which are carefully selected to maximize the capacity of the resulting channel. %
A natural motivation for this problem is that an input alphabet of controlled cardinality allows to control the complexity of the transmitter and receiver. Furthermore, when the channel transition probability function is unknown, the restriction to a subset of the input symbols may reduce the cost of estimating the effective transition probability function. This possibility is outlined in, e.g. \cite{egger2023maximal}, where the goal was to identify the maximal-capacity channel among a set of candidate channels, through adaptive exploration. If it can be determined during the exploration phase that capacity can be achieved without using some of the input symbols (while not knowing the capacity exactly at this stage), then this reduces the cost of accurately estimating the capacity during the rest of the exploration phase.

The problem of input selection has been studied for the special case of conditionally Gaussian channels in \cite{chan2005capacity}, and in the context of Multiple-Input Multiple-Output (MIMO) channels, e.g., \cite{konar2018simple}. In the latter, the authors show submodularity for the problem of antenna subset selection for MIMO. This is useful, since the submodularity property leads to theoretical guarantees on the capacity achieved by greedy algorithms. Nonetheless, as we show, an analogous submodularity property does not hold for the inputs of DMCs, and so does not lead to direct performance guarantees on greedy algorithms. In \cite{wesel2018efficient} the binomial channel was considered, whose input alphabet is the continuous interval $[0,1]$, and an efficient algorithm for finding the finitely supported capacity-achieving input distribution was proposed (called Dynamic Blahut-Arimoto). The algorithm was recently generalized to the multinomial channel in \cite{kobovich2023mdab}). 

The papers that consider DMCs, and hence, that are closest to ours, are \cite{mezghani2008achieving} and \cite{schmeink2011capacity}. The authors of \cite{mezghani2008achieving} stress that among different formulations of the problem, the standard formulation of capacity through the maximization of the mutual information is the most interesting from an information-theoretic perspective, but conclude that it is challenging to efficiently solve this formulation. Hence, they instead focus on optimizing the symbol error rate or the cut-off rate. The papers \cite{schmeink2010reducing,schmeink2011capacity} considered capacity, but also argued that achieving theoretical guarantees is  challenging, and so focused on numerical approaches based on the Blahut-Arimoto algorithm. 

In this work, we revisit the problem of selecting input symbols for capacity maximization, and take a principled intermediate approach between generic greedy optimization methods, and high-complexity exhaustive-search optimal methods. Based on the properties of the transition matrix of a DMC, we first derive bounds on the loss in capacity incurred by only using a selected subset of the input symbols for transmission. We then use this bound as a surrogate measure in designing an algorithm for input symbol selection, and show the effectiveness of the proposed algorithm in various scenarios. Interestingly, our algorithm operates without computing the original channel's capacity (with full use of input symbols). This is useful in case the input alphabet is very large and accurate computation of the capacity is computationally demanding.  

Informally, our algorithm is based on clustering of similar rows of the channel transition matrix. The novelty is in the choice of cluster representatives: Our upper bound depends on the subset of the output alphabet's probability simplex covered by the transition probabilities of the selected input symbols. Thus, the algorithm chooses the representatives to maximize this subset, and thus, to reduce the loss in capacity compared to the full usage of the input symbols. 
Such clustering points in the probability simplex have been studied, e.g., in \cite{nielsen2017clustering}, but our algorithm is tailored to maximize the mutual information, and hence differs from general-purpose choices.

\section{Problem Formulation}
\paragraph*{Notation}
Random variables are denoted by capital script letters $\rvZ$, their realizations by small letters, and sets by calligraphic letters $\mathcal{Z}$. The entropy of random variable $\rvZ$ over alphabet $\mathcal{Z}$ is denoted by $\entropy{\rvZ} \define -\sum_{z \in \mathcal{Z}} \Pr(\rvZ=z) \log(\Pr(\rvZ=z))$. All logarithms are taken to the natural base unless stated otherwise. The probability simplex over the alphabet $\mathcal{Z}$ is denoted by $\mathcal{P}(\mathcal{Z})$. The KL-divergence between distributions $P$ and $Q$ is denoted by $\kl{P}{Q}$, the $\chi^2$-divergence by $\chisq{P}{Q}$, and the Jensen-Shannon divergence by $\jsd{P}{Q} \define \frac{1}{2} \kl{P}{M} + \frac{1}{2} \kl{Q}{M},$ where $M = \frac{P + Q}{2}$. For an integer $\tau$,  $[\tau] \define \{1, \dots, \tau\}$. 

We consider a DMC $\channel$ with input alphabet $\symtx$ and output alphabet $\symrx$ to be an indexed set of its conditional probability mass functions $W=\{\pmfcondx\}_{x \in \symtx}$, or, alternatively, the transition matrix $\pmf$. When we use only a subset $\centers$ of the input symbols, we conveniently refer to the channel as $W_\centers \define \{\pmfcondx\}_{x \in \centers}$, for some $\centers \subset \symtx$. The mutual information between the input distribution $\inputdist$ and the output distribution $P_Y$ induced by the channel $\pmf$ is denoted by $\mutinf{\inputdist; \pmf}$.
The capacity of channel $\pmf$ can be written as the maximization of the mutual information over $\inputdist$, i.e.,
\begin{align*}
    C(\fullchannel) = \max_{\inputdist \in P(\symtx)} \mutinf{\inputdist; \pmf}.
\end{align*}
While the optimizer $\inputdist^\star$ to this optimization problem is not unique, any optimizer induces the unique capacity-achieving output distribution $Q_\rvY^\star$ \cite[Corollary 2, Thm 4.5.1]{gallager1968information}. Throughout the paper, we additionally make use of the dual problem of this optimization problem, also known as the minimax capacity theorem \cite{csiszar1972class,kemperman1974shannon}, which states that 
\begin{equation}
\mathrm{C}(\fullchannel)=\min_{Q_\rvY \in {\cal P}({\cal Y})}\max_{x\in{\cal X}}\kl{\pmfcondx}{ Q_\rvY},\label{eq: Csiszar minimax capacity}
\end{equation}
whose minimizer is the unique capacity-achieving output distribution $Q_\rvY^\star$. Furthermore, by the convexity of the optimization problem, the KL-divergence $\kl{\pmfcondx}{Q_\rvY^\star}$ for all input symbols $x \in \symtx$ equals $C(\fullchannel)$ if $\inputdist^\star(x)>0$ and is at most $C(\fullchannel)$ if $\inputdist^\star(x)=0$ \cite[Thm 4.5.1]{gallager1968information}. Hence, the capacity is the \textit{information radius} of the collection of conditional distributions.

We focus on the following questions: \textit{How to select a subset $\centers^* \subset \symtx$ of input symbols, such that $\inputdist^\star$ is supported on $\centers^*$, $|\centers^*| \leq k$ and the capacity is maximized among all other choice of $\centers$, $|\centers| \leq k$?  How to quantify the capacity loss compared to the channel that uses all symbols in $\symtx$, i.e., $C(W) - C(W_{\centers^*})$?} In order to demonstrate the difficulty of the problem, we next present two approaches that are commonly used to lower the complexity of such optimization problems, yet we show that both fail to achieve that (at least not in the direct manner that we have considered). 

First, we may consider a relaxation of the constraint. We note that the constraint in the primal formulation by limiting the input distribution $\inputdist$ to a support of size $\ncenters$, i.e.,
\begin{align*}
    C(\fullchannel) = \max_{\inputdist \in P(\symtx): \Vert \inputdist \Vert_0 = \ncenters} \mutinf{\inputdist; \pmf},
\end{align*}
where $\Vert x \Vert_0 \define \lim_{p\rightarrow 0} \sum_{i} \vert x \vert^p$ is known as the $L_0$-pseudonorm. This formulation limits the support of $\inputdist$ to $\ncenters$, but is non-concave and NP-hard in general \cite{feng2013complementarity}. A common approach is to relax the $L_0$ constraint to an $L_1$ constraint (or higher order). However, even such relaxed constraint still results in a non-concave optimization problem.

Second, we may consider showing that this problem is submodular, since this facilitates various optimization tools with theoretical guarantees  \cite{bach2013learning, bilmes2022submodularity}, and specifically, guarantees on the loss of greedy algorithms. To that end, recall the following definition of a submodular set function:
\begin{definition}[Submodular Set Functions] \label{def:submodularity}
    Consider a set $\tmpsetfull$, and let  $2^{\tmpsetfull}$ be its power set. A function $f: 2^{\tmpsetfull} \rightarrow \mathbb{R}$ is submodular if and only if for every $\tmpseta \subset \tmpsetb \subset \tmpsetfull$ and an element $\tmpsetitem \in \tmpsetfull \setminus \{\tmpsetb\}$, it  holds that
    \begin{align*}
        f(\tmpseta \cup \{\tmpsetitem\}) - f(\tmpseta) \geq f(\tmpsetb \cup \{\tmpsetitem\}) - f(\tmpsetb).
    \end{align*}
\end{definition}
This version of the definition of submodularity is based on the \emph{diminishing returns} property. Informally, adding an element to a small set $\tmpseta$ increases the value function by at least as much as adding the same element to a larger superset $\tmpsetb$. 
However, we have the following:
\begin{observation}
    The capacity of a DMC is \textit{not} submodular in the set of input symbols.
\end{observation}
We show this observation %
via  two counterexamples. First, trivially, the diminishing return property breaks down when $\tmpseta=\emptyset$, since the capacity can never increase by adding any symbol to an empty set. Second, assume that the conditional entropies $\entropy{\pmfcondx}$ are equal for all $x \in \symtx$ and, hence, capacity is achieved by maximizing $\entropy{\rvY}$ through balancing the output distribution. Indeed, suppose that there are $4$ input symbols, such that the first and second symbols (resp. third and fourth) are \textit{complementary}, in the sense that a uniform mixture of their corresponding rows is a uniform distribution in $\mathcal{P}(\symrx)$ (or some other high-entropy distribution).  In this case, given the first two symbols, adding the third symbol will "unbalance" the output distribution and reduce its entropy, and thus the mutual information, while adding the fourth symbol will "re-balance" it, and again, increase that entropy. Evidently, this contradicts the submodularity condition.

Consequently, submodular properties and guarantees for greedy optimization algorithms cannot be directly exploited for the problem of channel input symbol selection. %
We note in passing that the capacity of a DMC may still fulfill \textit{approximate} notions of submodularity \cite{feige2017approximate, abhimanyu2018approximate, chierichetti2022additive},  but we leave an investigation of this possibility for further research. %

\section{Theoretical Guarantees on Input Symbol Selection}
In this section, we present our main theoretical guarantees on the selection of input symbols $\centers$ for DMCs that maximize the capacity of the resulting channel. %
When exploring the importance of input symbols for maximizing the capacity of a channel, it is natural to examine the interplay between the conditional distributions given by the rows of the channel's transition matrix. We will concentrate our analysis on the convex hull spanned by a certain subset of the symbols' conditional distributions, i.e., the transition matrix generated by $\centers$, and its relation to unused symbols' distributions. Depending on the shape of the channel, such symbols $\symtx \setminus \centers$ can potentially be pruned without or with only a minor loss in the capacity of the channel. We next formally define a convex hull of a channel based on a subset of the symbols $\centers \subset \symtx$:
\begin{definition}[Convex Hull of a Channel]
    Let $\centers$ be a set of symbols that form a channel $W_\centers \define \{\pmfcondx[r]\}_{r \in \centers}$. The convex hull of the channel is defined as %
    \begin{align*}
        \hull{W_\centers} \define \left\{\sum_{r \in \centers} c_r \pmfcondx[r]\right\}_{\mathbf{c} \in \mathcal{P}(\mathcal{R})} \subseteq \mathcal{P}(\symrx) . %
    \end{align*}
\end{definition}

Having this definition at hand, we start with the special case of selecting input symbols when the input alphabet $\symtx$ is large compared to the output alphabet $\symrx$. We later move to the cases in which the alphabet sizes are of the same order, and investigate when input selection can be done without a loss in capacity, and how to bound the loss in capacity otherwise.

\subsection{Symbol Selection Without Capacity Loss}
When the input alphabet of a DMC is large compared to its output alphabet, the convex hull of the channel may be the entire output simplex. Indeed, it is well known that, due to Carathéodory's theorem, the capacity-achieving output distribution $Q_\rvY^\star \in \hull{W}$ can be written as a convex combination of at most $\card{\symrx}$ extreme points (conditional distributions) corresponding to inputs  $\extremalpoints \subset \symtx$. Thus, independently of the transition matrix, there exists a set of input symbols $\extremalpoints$ with cardinality at most $\card{\symrx}$, which achieves the capacity \cite[Corollary 3, Thm 4.5.1]{gallager1968information}. 

Even if the number of input symbols $\card{\symtx}$ of a DMC is at most the size of the output alphabet $\card{\symrx}$, we can potentially utilize the properties of the convex hull of a channel to prune some of the input symbols without losing in capacity. It can be found that this applies when the conditional distribution of a symbol $x$ lies within the convex hull of the channel spanned by a subset of the symbols $\centers$. In general, symbol $x \in \symtx$ can be removed from $\symtx$ without loss in capacity when $\pmfcondx[\allidx]$ is a convex combination of the channels of the remaining symbols. We formalize and generalize this statement as follows:

\begin{proposition} \label{thm:inconvexhull}
    Consider a channel $W = \{\pmfcondx[\allidx]\}_{\allidx \in \allsym}$, where the input symbols $\allsym$ are partitioned into two disjoint sets $\redsym$ and $\spanningsym$, such that the conditional distributions $\pmfcondx[\redidx]$ of the symbols in $\redidx \in \redsym$ are contained in the convex hull of conditional distributions of the remaining symbols in $\spanningidx \in \spanningsym$. Then, the symbols in $\redsym$ can be removed from the input alphabet without incurring a loss in capacity, that is,
    \begin{align*}
        C(\{\pmfcondx[\allidx]\}_{\allidx \in \allsym}) = C(\{\pmfcondx[\spanningidx]\}_{\spanningidx \in \spanningsym}).
    \end{align*}
    Hence, there exists a capacity-achieving input distribution for which $\inputdist(x)=0$ for $x \in \symtx \setminus \centers$ and $\inputdist(x) \geq 0$ for $x \in \centers$.
\end{proposition}

Consequently, keeping those symbols in the channel that form the convex hull suffices to achieve the capacity. Formally, let $W$ be a channel over the input alphabet $\symtx$. There exists a capacity-achieving input distribution supported only on the input symbols that span the convex hull of $W$. %

\subsection{Bounding the Capacity Loss}

Removing symbols from the input alphabet that span the convex hull of a DMC will likely lead to a loss in capacity. Note that this might not be the case; e.g., in the case where the input alphabet is larger than the output alphabet, we can always write the capacity-achieving output distribution as convex combination of at most $\card{\symrx}$ extreme points, in this case the input symbols' conditional distributions. Those points do not necessarily span the entire convex hull of the full channel. Hence, this special case is not covered by \cref{thm:inconvexhull}, yet does not lead to a capacity loss. %

By knowing the distance from the conditional distributions associated with symbols in $\centers$ to those in $\symtx \setminus \centers$, one can bound the capacity loss incurred by restricting the input distribution to be supported only on the symbols in $\centers$ instead of the entire alphabet $\symtx$. 
This notion is captured by the following natural concept of the nearest neighbor of a symbol $x$ to another set of symbols $\centers$, which we define using the $\chi^2$-divergence. The usage of $\chi^2$-divergence stems from our bound in \cref{thm:closetohull}.

\begin{definition}[Nearest Neighbor]
    The nearest neighbor $\nearestcenter$ of a symbol $x \in \symtx \setminus \centers$ is the symbol $r \in \centers$ closest to $x$ in terms of the $\chi^2$-divergence, i.e., $$\nearestcenter \define \argmin_{r \in \centers} \chisq{\pmfcondx[x]}{\pmfcondx[r]}.$$
\end{definition}

Nonetheless, in the context of capacity maximization,  it is sufficient to consider the distance between each of the removed symbols $\symtx \setminus \centers$ and the convex hull $\hull{W_\centers}$. In particular, the $\chi^2$ distance between the distributions will lead to theoretical guarantees for capacity. We formally define the distance in the following.

\begin{definition}[Distance to the Convex Hull of a Channel]
    Let $W_\centers \define \{\pmfcondx[r]\}_{r \in \centers}$ be a channel that generates a convex hull on the probability simplex. Then, for a symbol $x \in \symtx~\setminus~\centers$ whose conditional distribution $\pmfcondx$ is not contained in $\hull{W_\centers}$, assuming the convex hull spans at least one distribution with the same support as $\pmfcondx$, the distance from $\pmfcondx$ to the convex hull of $W_\centers$ is defined as
    \begin{align*}
        \hulldev{\centers}{x} \define \min_{\closestconditionalvar[x] \in \hull{W_\centers}} \chisq{\pmfcondx}{\closestconditionalvar[x]}.
    \end{align*}
    Let $\closestconditional[x]$ be the distribution that minimizes the above objective so that $\hulldev{\centers}{x} = \chisq{\pmfcondx}{\closestconditional[x]}.$
\end{definition}

With those definitions at hand, we can bound the expected loss in capacity by removing a set of symbols $\pruningsym$ from the alphabet $\symtx$. Let $\centers = \symtx \setminus \pruningsym$ be the selected (remaining) symbols. Then we can divide the set $\pruningsym$ of removed (unused) symbols into symbols within the convex hull of the channel $W_\centers$ (referred to as $\inconvexhullsym$) and symbols outside the convex hull (referred to as $\notinconvexhullsym$), i.e., $\pruningsym = \inconvexhullsym \cup \notinconvexhullsym$. From \cref{thm:inconvexhull}, we know that not using symbols from $\inconvexhullsym$ will not decrease the channel's capacity. What remains is to quantify the loss in capacity when removing symbols not contained in the convex hull of the remaining ones. We establish such a result in \cref{thm:closetohull}. For the proof and the theorem statement, we rely on the following concept of a pseudo-simplex and pseudo-capacity, as introduced in \cite{egger2023maximal}.

\begin{definition}[Pseudo-Simplex and Pseudo-Capacity]
    Let $\mathcal{P}_\eta(\mathcal{Y})$ be a subset of the probability simplex $\mathcal{P}(\mathcal{Y})$ over alphabet $\symrx$, where each probability mass is at least $\eta$. The pseudo-capacity of a DMC $W$ with input symbols $\symtx$ is
    \begin{align*}
        C_\eta(W) = \min_{Q_\rvY \in \mathcal{P}_\eta(\mathcal{Y})} \max_{x \in \symtx} \kl{\pmfcondx[x]}{Q_\rvY}.
    \end{align*}
    Let $\pseudocaod(W) \define \argmin_{Q_\rvY \in \mathcal{P}_\eta(\mathcal{Y})} \max_{x \in \symtx} \kl{\pmfcondx[x]}{Q_\rvY}$ be the capacity achieving output distribution that attains the pseudo-capacity; hence minimizes the above quantity.
\end{definition}

For the statement of the theorem, we define for a symbol $x \in \symtx \setminus \mathcal{R}$ the quantity $\lbp$ as the smallest non-zero probability mass of the distance-minimizing distribution $\closestconditional[x]$, i.e.,
\begin{align*}
    \lbp &\define \min_{y \in \symrx: \closestconditional[x](y) \neq 0} \closestconditional[x](y). %
\end{align*}
It should be noted that $\lbp$ and $\closestconditional[x]$ depend on $\centers$, but we omit this from the notation for brevity.
\begin{theorem} \label{thm:closetohull}
    For some $0 < \eta \leq \frac{1}{2\card{\symrx}}$, a set of chosen symbols $\centers \subset \symtx$ and any $x \in \symtx$ that maximizes $\kl{\pmfcondx[x]}{\pseudocaod}$, with $\hulldev{\centers}{x} \geq 0, \lbp > 0$, the capacity loss due to only using inputs  in $\centers$ is bounded as
    \begin{align*}
    C(\fullchannel) \!-\! C(\prunedchannel) 
    &\leq 4\card{\symrx} \eta + \hulldev{\centers}{x} \\
    &+ \sqrt{-\log(\eta) \left( C(\prunedchannel) \!-\! \log \left(\lbp \right)\right)} \sqrt{\hulldev{\centers}{x}}. %
\end{align*}
For multiple maximizers, $x$ can be chosen to minimize the upper bound.
The parameter $\eta$ exhibits a bias-variance trade-off. With $x^\prime \in \symtx$ being the maximizer of $\kl{\pmfcondx[x]}{Q_\rvY^\star}$ among all $x\in \symtx$, a suitable choice of $\eta$ is $$\eta = \left(\frac{\sqrt{C(\prunedchannel) - \log \left(\lbp[x^\prime] \right)} \sqrt{\hulldev{\centers}{x^\prime}}}{4\card{\symrx}} + 0.07 \right)^2.$$

\end{theorem}

A simplified, yet looser, statement is obtained when considering for symbol $x$ the nearest neighbor $\nearestcenter$ instead of $\hulldev{\centers}{x}$, which avoids computing the distance to the convex hull.

\begin{proof}[Sketch of the Proof]
We apply \cref{thm:inconvexhull} and the triangle inequality to show the equivalence of comparing the capacity of $W$ to either (i) the capacity of $\prunedchannel$ or (ii) the capacity of  $\prunedchannel^\prime$, which is the channel that additionally contains all distributions $\closestconditional[x]$ in the convex hull of $\prunedchannel$ that represent the closest to each symbol $x\in \symtx \setminus \centers$ in terms of $\chi^2$-distance. Hence, we ensure that the nearest neighbor $\nearestcenter$ of each symbol $x$ in $W$ is the distribution in the convex hull of $\prunedchannel$ that represents the minimum distance $\hulldev{\centers}{x}$. Next, we use pseudo-capacity to bound the difference between $C_\eta(W)$ and $C_\eta(\prunedchannel)$ based on the distance to the nearest neighbors (and due to the above, equivalently the distance to the convex hull) of certain symbols. Therefore, we make use of the following lemma, for which we restrict the capacity-achieving output distribution $Q_\rvY$ to the pseudo-simplex $\mathcal{P}_\eta(\mathcal{Y})$.
\begin{lemma} \label{lemma:pseudo_capacity_diff}
     For any choice of $0 < \eta \leq \frac{1}{2\card{\symrx}}$, let $\pseudocaod$ be as defined above. Then, with $x$ being the maximizer of $\kl{\pmfcondx[x]}{\pseudocaod}$ and $\tmprep[x] \in \centers$ its nearest neighbor that minimizes $\chisqtmp \define \chisq{\pmfcondx[x]}{\pmfcondx[\tmprep]}$, we have the following upper bound
    \begin{align*}
        &C_\eta\left( \fullchannel \right) - C_\eta\left(\prunedchannel \right) \\
        &\leq \chisqtmp %
        + \sqrt{\chisqtmp} \sqrt{C(\prunedchannel) \log\left(\frac{\lbp}{\eta}\right) + \log\left(\frac{1}{\eta} \right)\log\left(\frac{1}{\lbp} \right)}.   %
    \end{align*}
\end{lemma}
The gap between $C_\eta(\prunedchannel)$ to capacity $C(\prunedchannel)$ (and between $C_\eta(W)$ and $C(W)$, respectively) can be bounded by $2\eta\card{\symrx}$ by the application of \cite[Lemma 1]{egger2024maximal}. With an appropriate choice of $\eta$ to trade off the linear and non-linear terms, we obtain the statement in the theorem. Note that the computation of $\kl{\pmfcondx[x]}{\pseudocaod}$ determines the choice of $x$ for which the bound applies; hence, it requires a choice of $\eta$. To find a value for $\eta$ that is well suited for the maximum of $\kl{\pmfcondx[x]}{\pseudocaod}$, we use the surrogate objective $\kl{\pmfcondx[x]}{Q_\rvY^\star}$ for the calculation of $\eta$.
\end{proof}

\section{Input Symbol Selection by Clustering}

We now turn our attention to designing a practical algorithm for selecting a subset of $\ncenters$ input symbols that minimize the loss in capacity compared to using all possible symbols.  For this problem, exhaustive search over all possible subsets is typically computationally infeasible,  especially in cases where $\ncenters$ and the input alphabet are large. %
On the other hand, greedy algorithms require computing the channel's capacity for many candidate sets of symbols, which can be computationally demanding, especially for large $\ncenters$, and might produce sub-optimal solutions. Further, as discussed, there is no simple theoretical guarantee on the loss in capacity incurred by the solutions of greedy algorithms compared to the optimal solution. We thus propose a clustering-based algorithm, which, in the first step, clusters symbols with similar conditional distributions, and in the second step, carefully chooses representatives of each cluster. Our proposed algorithm does not require any possibly expensive capacity calculation. %
We compare its performance to the optimal solution (whenever finding it with an exhaustive search is feasible).

Our strategy, summarized in~\cref{alg:cap}, first partitions the input symbols according to their conditional distributions into $\ncenters$ clusters of similar symbols. Then, having identified the clusters, it determines a representative for each cluster. Inspired by \cref{thm:closetohull}, the representatives are chosen such that their resulting convex hull likely contains as many removed symbols as possible, while simultaneously minimizing the distance to the symbols outside the convex hull.

We use agglomerative (hierarchical) clustering due to its simplicity and flexibility in the usage of distance measures. The clustering first assigns one cluster to each symbol. %
Then, the pairwise Jensen-Shannon divergence between all symbols is computed. The pairwise distance between two clusters, called \emph{linkage}, is computed as the \emph{maximum} distance between their respective elements, i.e., for two clusters $m$ and $n$ with elements $\cluster_m$ and $\cluster_n$, respectively, we have $\linkage{m}{n} \define \max_{x \in \cluster_m, x^\prime \in \cluster_n} \jsd{\pmfcondx}{\pmfcondx[x^\prime]}$. The two clusters with the minimum linkage are merged. This process is repeated until $\ncenters$ clusters remain. While the bound in \cref{thm:closetohull} suggests the $\chi^2$-divergence as a distance measure, we chose the Jensen-Shannon divergence and the linkage as the maximum distance between cluster elements since these achieved the best empirical results. It is of interest to close this gap and propose bounds based on JSD instead.

To choose a representative for each cluster, we select the symbol $x_m$ of cluster $m$ that maximizes the average distance to all other symbols $x \neq x_m \in \symtx$. We found that among various approaches, selecting the representative as the symbol in a cluster whose distribution maximizes the average distance to all other symbols (including those outside the cluster) provides a good trade-off between compute cost and the resulting performance. This aligns with \cref{thm:closetohull}, that advocates the  selection of representatives that create a large convex hull.

\cref{fig:dmc_large} shows an illustrative example of the performance. It compares our hierarchical clustering algorithm with the bounds of \cref{thm:closetohull}, for different values of $\ncenters \in \{2, \dots, 10\}$, for a DMC with input and output alphabet sizes $\card{\symtx} = \card{\symrx} = 30$. The DMC is generated according to a Dirichlet distribution, as explained in the sequel. The capacity of the full channel is computed by including all input symbols in $\symtx$. As a baseline for our clustering algorithm, we include the results from an exhaustive search over all sets of $k$ symbols. %

\begin{algorithm}[t]
\caption{Input Symbol Selection.}\label{alg:cap}
\begin{algorithmic}
\Require DMC $W$, desired number $\ncenters < \card{\symtx}$ of inputs
\State Compute $\jsd{\pmfcondx}{\pmfcondx[x^\prime]}$  for all symbols $x, x^\prime$
\State Initialize clusters $\cluster_m, m \in [\card{\symtx}]$, one for each $\pmfcondx[x]$
\While{Number of clusters $> \ncenters$}
\State For all clusters $m,n$, compute the pairwise linkage
\State $\displaystyle \linkage{m}{n} = \max_{x \in \cluster_m, x^\prime \in \cluster_n} \jsd{\pmfcondx}{\pmfcondx[x^\prime]}$ 
\State Merge clusters $m^\star,n^\star$ with minimum $\linkage{m}{n}$, i.e., 
\State $\cluster_{m^\star}$ = $\cluster_{m^\star} \cup \cluster_{n^\star}$
\EndWhile
\For{Each remaining cluster $m\in [\ncenters]$}
    \State Select the symbol $x_m \in \cluster_m$ of cluster $m$ maximizing 
    \State average distance to all other symbols $x \neq x_m$ in $\symtx$
\EndFor
\State Output $W^\prime=\{\pmfcondx[x_m]\}_{m \in [\ncenters]}$
\end{algorithmic}
\end{algorithm}

For $\eta$, we use the value proposed in \cref{thm:closetohull}. For values $\ncenters < 5$, this choice of $\eta$ exceeds the limit of $\frac{1}{2\card{\symrx}}$. Hence, obtaining a tight bound is not possible. We plot the optimal solution obtained through an exhaustive search for small values of $\ncenters$, and the capacity of the full channel that uses all the input symbols in $\symtx$. It can be seen that our clustering algorithm, on average, finds an optimal selection of the input symbols. Even without knowing the capacity of the full channel, with $\ncenters=5$, the bound from \cref{thm:closetohull} indicates that accounting for the remaining $25$ symbols cannot improve capacity by more than $\approx 0.8$ bit. In this simple example the bound from \cref{thm:closetohull} is conservative (as can be seen from the exact value of capacity). However, for very large input alphabet channels, the capacity is infeasible to compute, and this bound may be the only indication. %

\begin{figure}[tb]
    \centering
    \resizebox{\linewidth}{!}{\begin{tikzpicture}

\definecolor{darkgray176}{RGB}{176,176,176}
\definecolor{darkorange25512714}{RGB}{255,127,14}
\definecolor{lightgray204}{RGB}{204,204,204}
\definecolor{steelblue31119180}{RGB}{31,119,180}

\begin{axis}[
legend cell align={left},
legend style={
  fill opacity=0.8,
  draw opacity=1,
  text opacity=1,
  at={(0.97,0.03)},
  anchor=south east,
  draw=lightgray204
},
tick align=outside,
tick pos=left,
unbounded coords=jump,
x grid style={darkgray176},
xlabel={Numer of Symbols \(\displaystyle k\)},
xmajorgrids,
xmin=1.6, xmax=10.4,
xtick style={color=black},
y grid style={darkgray176},
ylabel={\(\displaystyle C(W)\) in Bits},
ymajorgrids,
ymin=-0.25, ymax=3.25,
ytick style={color=black}
]
\path [draw=steelblue31119180, fill=steelblue31119180, opacity=0.2]
(axis cs:5,3.14340185422666)
--(axis cs:5,2.84534471862944)
--(axis cs:6,2.83590132494043)
--(axis cs:7,2.83676958576267)
--(axis cs:8,2.84344623075994)
--(axis cs:9,2.7995624144448)
--(axis cs:10,2.77508573706428)
--(axis cs:10,3.09032318172761)
--(axis cs:10,3.09032318172761)
--(axis cs:9,3.14667771331086)
--(axis cs:8,3.15084415890816)
--(axis cs:7,3.1580979758117)
--(axis cs:6,3.13843607774373)
--(axis cs:5,3.14340185422666)
--cycle;

\path [draw=darkorange25512714, fill=darkorange25512714, opacity=0.1]
(axis cs:2,2.32128961680119)
--(axis cs:2,2.07313817711257)
--(axis cs:3,2.07313817711257)
--(axis cs:4,2.07313817711257)
--(axis cs:5,2.07313817711257)
--(axis cs:5,2.32128961680119)
--(axis cs:5,2.32128961680119)
--(axis cs:4,2.32128961680119)
--(axis cs:3,2.32128961680119)
--(axis cs:2,2.32128961680119)
--cycle;

\path [draw=black, fill=black, opacity=0.4]
(axis cs:2,0.99173235909128)
--(axis cs:2,0.990838712205404)
--(axis cs:3,1.56850728351366)
--(axis cs:4,1.94442615601366)
--(axis cs:5,2.07313822992655)
--(axis cs:5,2.32128965005411)
--(axis cs:5,2.32128965005411)
--(axis cs:4,1.99431744288143)
--(axis cs:3,1.57068481113657)
--(axis cs:2,0.99173235909128)
--cycle;

\path [draw=steelblue31119180, fill=steelblue31119180, opacity=0.2]
(axis cs:2,0.991680032935819)
--(axis cs:2,0.990253333097961)
--(axis cs:3,1.56818728785807)
--(axis cs:4,1.94442615601366)
--(axis cs:5,2.0731382290448)
--(axis cs:6,2.07313822224781)
--(axis cs:7,2.0731382173207)
--(axis cs:8,2.0731382109899)
--(axis cs:9,2.07313820934852)
--(axis cs:10,2.07313820280689)
--(axis cs:10,2.32128962812745)
--(axis cs:10,2.32128962812745)
--(axis cs:9,2.32128963127817)
--(axis cs:8,2.32128963424631)
--(axis cs:7,2.32128963892255)
--(axis cs:6,2.32128964202996)
--(axis cs:5,2.32128965032784)
--(axis cs:4,1.99431744288143)
--(axis cs:3,1.57063004431308)
--(axis cs:2,0.991680032935819)
--cycle;

\addplot [line width=1.2pt, steelblue31119180, dashed]
table {%
2 nan
3 nan
4 nan
5 2.99437328642805
6 2.98716870134208
7 2.99743378078718
8 2.99714519483405
9 2.97312006387783
10 2.93270445939595
};
\addlegendentry{Upper Bound (cf. \cref{thm:closetohull})}
\addplot [line width=1.2pt, darkorange25512714, dashed]
table {%
2 2.19721389695688
3 2.19721389695688
4 2.19721389695688
5 2.19721389695688
6 2.19721389695688
7 2.19721389695688
8 2.19721389695688
9 2.19721389695688
10 2.19721389695688
};
\addlegendentry{Full Channel (no symbol removal)}
\addplot [line width=1.2pt, steelblue31119180]
table {%
2 0.99096668301689
3 1.56940866608557
4 1.96937179944755
5 2.19721393968632
6 2.19721393213889
7 2.19721392812163
8 2.19721392261811
9 2.19721392031334
10 2.19721391546717
};
\addlegendentry{Hierarchical Clustering (cf. \cref{alg:cap})}
\addplot [black, line width=1.2pt, dotted, mark=x, mark size=5, mark options={solid}]
table {%
2 0.991285535648342
3 1.56959604732512
4 1.96937179944755
5 2.19721393999033
};
\addlegendentry{Exhaustive Search}
\end{axis}

\end{tikzpicture}}
    \caption{Input selection results for $50$ DMCs with input and output alphabet size $\card{\symtx} = \card{\symrx} = 30$, randomly sampled with $\nclustersdir=5$, $\dirscalea = 0.005$ and $\dirscaleb=10^{10}$. Lines show average results, shaded areas the standard deviation.}
    \label{fig:dmc_large}
\end{figure}

\emph{Generating a DMC with the Dirichlet distribution:} %
We introduce a random, yet structured, \emph{hierarchical sampling} method for generating DMCs with input alphabet $\symtx$.
We use the Dirichlet distribution parameterized by a vector $\alpha = (\alpha_1, \dots, \alpha_{\card{\symrx}})$ which generates a probability distribution of the same dimension, i.e., over $\mathcal{P}(\mathcal{Y})$. The relation $\frac{\alpha_y}{\sum_{y \in \symrx} \alpha_y}$ determines the average probability mass of the symbols $y \in \symrx$. As opposed to uniform or Gaussian sampling of the conditional distributions, Dirichlet sampling allows to  tune the expected capacity of the channels through the parameter choice. The larger the values $\alpha_y, y\in \symrx$, the more noisy the rows of the transition matrices will be; whereas smaller values of $\alpha_y$ will lead to cleaner rows of the transition matrices.\footnote{We add a small non-negative mass to each entry and normalize to ensure that the convex hull has the same support as all the removed rows, as assumed for \cref{thm:closetohull}. Meeting this assumption is likely in practice.}

To model that the rows of the transition matrices might be dependent, we first sample  $\nclustersdir < \card{\symtx}$ of Dirichlet samples with $\alpha = \dirscalea (1, \dots, 1)$, where the parameter $\dirscalea$ will determine the variance of the sample, hence the capacity of the channel. Each sample is then used as a parameterization for a new Dirichlet distribution, scaled by $\dirscaleb$. Hence, those distributions will be noisy samples around the previously sampled ones, which justifies the term hierarchical sampling. %
For our simulations, we used $\nclustersdir=5$, $\dirscalea = 0.005$ and $\dirscaleb=10^{10}$.
\section{Conclusion}

We investigated the problem of capacity-optimal input symbol selection for DMCs. Based on the channel's transition matrix, we derived bounds for the capacity loss incurred by the removal of specific input symbols. We showed the dependency of the bounds on the $\chi^2$-distance of the removed symbols to the convex hull spanned by the selected inputs. We transferred our theoretical results to designing a clustering-based selection algorithm with a cluster representative choice tailored to maximizing the size of the resulting convex hull. With DMCs randomly sampled with a Dirichlet distribution, we compared our algorithm to an exhaustive search and observed the established theoretical bound on the capacity loss. %

\bibliographystyle{IEEEtran}
\bibliography{references}

\begin{thebibliography}{10}
\providecommand{\url}[1]{#1}
\csname url@samestyle\endcsname
\providecommand{\newblock}{\relax}
\providecommand{\bibinfo}[2]{#2}
\providecommand{\BIBentrySTDinterwordspacing}{\spaceskip=0pt\relax}
\providecommand{\BIBentryALTinterwordstretchfactor}{4}
\providecommand{\BIBentryALTinterwordspacing}{\spaceskip=\fontdimen2\font plus
\BIBentryALTinterwordstretchfactor\fontdimen3\font minus
  \fontdimen4\font\relax}
\providecommand{\BIBforeignlanguage}[2]{{%
\expandafter\ifx\csname l@#1\endcsname\relax
\typeout{** WARNING: IEEEtran.bst: No hyphenation pattern has been}%
\typeout{** loaded for the language `#1'. Using the pattern for}%
\typeout{** the default language instead.}%
\else
\language=\csname l@#1\endcsname
\fi
#2}}
\providecommand{\BIBdecl}{\relax}
\BIBdecl

\bibitem{egger2023maximal}
M.~Egger, R.~Bitar, A.~Wachter-Zeh, D.~Gündüz, and N.~Weinberger,
  ``Maximal-capacity discrete memoryless channel identification,'' in
  \emph{IEEE International Symposium on Information Theory (ISIT)}, 2023, pp.
  2248--2253.

\bibitem{chan2005capacity}
T.~Chan, S.~Hranilovic, and F.~Kschischang, ``Capacity-achieving probability
  measure for conditionally {G}aussian channels with bounded inputs,''
  \emph{IEEE Transactions on Information Theory}, vol.~51, no.~6, pp.
  2073--2088, 2005.

\bibitem{konar2018simple}
A.~Konar and N.~D. Sidiropoulos, ``A simple and effective approach for transmit
  antenna selection in multiuser massive mimo leveraging submodularity,''
  \emph{IEEE Transactions on Signal Processing}, vol.~66, no.~18, pp.
  4869--4883, 2018.

\bibitem{wesel2018efficient}
R.~D. Wesel, E.~E. Wesel, L.~Vandenberghe, C.~Komninakis, and M.~Medard,
  ``Efficient binomial channel capacity computation with an application to
  molecular communication,'' in \emph{Information Theory and Applications
  Workshop (ITA)}, 2018, pp. 1--5.

\bibitem{kobovich2023mdab}
A.~Kobovich, E.~Yaakobi, and N.~Weinberger, ``{M-DAB}: An input-distribution
  optimization algorithm for composite {DNA} storage by the multinomial
  channel,'' \emph{arXiv preprint arXiv:2309.17193}, 2023.

\bibitem{mezghani2008achieving}
A.~Mezghani, M.~T. Ivrlac, and J.~A. Nossek, ``Achieving near-capacity on large
  discrete memoryless channels with uniform distributed selected input,'' in
  \emph{International Symposium on Information Theory and Its Applications},
  2008, pp. 1--6.

\bibitem{schmeink2011capacity}
A.~Schmeink and H.~Zhang, ``Capacity-achieving probability measure for a
  reduced number of signaling points,'' \emph{Wireless Networks}, vol.~17, pp.
  987--999, 2011.

\bibitem{schmeink2010reducing}
A.~Schmeink, R.~Mathar, and H.~Zhang, ``Reducing the number of signaling points
  keeping capacity and cutoff rate high,'' \emph{International Symposium on
  Wireless Communication Systems}, pp. 932--936, 2010.

\bibitem{nielsen2017clustering}
F.~Nielsen and K.~Sun, ``Clustering in hilbert simplex geometry,'' \emph{arXiv
  preprint arXiv:1704.00454}, 2017.

\bibitem{gallager1968information}
R.~G. Gallager, \emph{Information theory and reliable communication}.\hskip 1em
  plus 0.5em minus 0.4em\relax Springer, 1968, vol. 588.

\bibitem{csiszar1972class}
I.~Csiszár, ``A class of measures of informativity of observation channels,''
  \emph{Periodica Mathematica Hungarica}, vol.~2, no.~1, pp. 191--213, Mar.
  1972.

\bibitem{kemperman1974shannon}
J.~H.~B. Kemperman, ``On the {Shannon} capacity of an arbitrary channel,''
  \emph{Indagationes Mathematicae (Proceedings)}, vol.~77, no.~2, pp. 101--115,
  1974.

\bibitem{feng2013complementarity}
M.~Feng, J.~E. Mitchell, J.-S. Pang, X.~Shen, and A.~W{\"a}chter,
  ``Complementarity formulations of l0-norm optimization problems,''
  \emph{Industrial Engineering and Management Sciences. Technical Report.
  Northwestern University, Evanston, IL, USA}, vol.~5, 2013.

\bibitem{bach2013learning}
F.~Bach \emph{et~al.}, ``Learning with submodular functions: A convex
  optimization perspective,'' \emph{Foundations and Trends{\textregistered} in
  machine learning}, vol.~6, no. 2-3, pp. 145--373, 2013.

\bibitem{bilmes2022submodularity}
J.~Bilmes, ``Submodularity in machine learning and artificial intelligence,''
  \emph{arXiv preprint arXiv:2202.00132}, 2022.

\bibitem{feige2017approximate}
U.~Feige, M.~Feldman, and I.~Talgam-Cohen, ``Approximate modularity
  revisited,'' in \emph{Annual ACM SIGACT Symposium on Theory of Computing},
  2017, pp. 1028--1041.

\bibitem{abhimanyu2018approximate}
A.~Das and D.~Kempe, ``Approximate submodularity and its applications: Subset
  selection, sparse approximation and dictionary selection,'' \emph{Journal of
  Machine Learning Research}, vol.~19, no.~3, pp. 1--34, 2018.

\bibitem{chierichetti2022additive}
F.~Chierichetti, A.~Dasgupta, and R.~Kumar, ``On additive approximate
  submodularity,'' \emph{Theoretical Computer Science}, vol. 922, pp. 346--360,
  2022.

\bibitem{egger2024maximal}
M.~Egger, R.~Bitar, A.~Wachter-Zeh, D.~G{\"u}nd{\"u}z, and N.~Weinberger,
  ``Maximal-capacity discrete memoryless channel identification,'' \emph{arXiv
  preprint arXiv:2401.10204}, 2024.

\end{thebibliography}

\appendix
\subsection{Counter-Examples for Submodularity of a DMC} \label{app:submodularity}
Let $f$ be the capacity $C$ that maps $2^W$ to the capacity $0 \leq C(W) \leq \log(\card{W})$, where $W$ refers to the collection of conditional distributions $\{\pmfcondx\}_{x \in \symtx}$. The definition of diminishing returns breaks when considering the empty set $\tmpseta = \emptyset$ and a single-element set $\tmpsetb = \{\pmfcondx[x]\}$ such that $\tmpsetitem = \pmfcondx[x^\prime] \neq \pmfcondx[x]$ for any two distinct elements $x, x^\prime \in \symtx$. In this case, we have
\begin{align*}
    f(\tmpseta \cup \{\tmpsetitem\}) - f(\tmpseta) &= 0 - 0, \\
    f(\tmpsetb \cup \{\tmpsetitem\}) - f(\tmpsetb) &= \bar{C} - 0, %
\end{align*}
for some $0 < \bar{C} \leq \log(2)$, which contradicts the above definition of submodularity. While this result is straightforward, it remains to show that DMCs do not fulfill the diminishing return property when $\card{\tmpseta} \geq 1$ and $\card{\tmpsetb} \geq 2$. We show this through a counterexample. Consider the channel $W=\{\pmfcondx[x]\}_{x\in[4]}$:
\begin{align*}
    \pmfcondx[1] &= [0.6, 0.2, 0.1, 0.1], &
    \pmfcondx[2] &= [0.6, 0.1, 0.1, 0.2] \\
    \pmfcondx[3] &= [0.6, 0.1, 0.2, 0.1], &
    \pmfcondx[4] &= [0.1, 0.6, 0.1, 0.2]
\end{align*}
With a slight abuse of notation, let $C(\tmpseta)$ denote the capacity of the channel $\pmf$ with input symbols $\tmpseta \subseteq \symtx$. Then, we have that 
\begin{align*}
    C(\{1,2,4\}& \cup \{3\}) - C(\{1,2,4\} ) \\
    &> C(\{1,2\}  \cup \{3\} ) - C(\{1,2\}),
\end{align*}
which violates \cref{def:submodularity}.  This counterexample is constructed by noting that the quantities $\entropy{\pmfcondx}$ are equal for all $x \in \symtx$. Hence, capacity is obtained by maximizing the entropy of the output distribution, i.e.,
\begin{equation}
    C(W) = -\entropy{\rvY|\rvX=x} + \max_{\inputdist \in \mathcal{P}(\mathcal{X})} \entropy{\rvY}. 
\end{equation}
For such channels,  the question of diminishing returns can be answered by studying how well adding a certain input symbol to the sets $\tmpseta$ and $\tmpsetb$ balances out the capacity achieving output distribution, reflected by an increase in its entropy. Our example shows that adding a specific symbol to a large alphabet balances the output distribution more positively than for a smaller alphabet, thus violating sub-additivity.

\subsection{Proof of \cref{thm:inconvexhull}}

\begin{proof} 

Let $\redsym$ be symbols that lie in the convex hull of a set of symbols $\spanningsym$, i.e. each symbol $\redidx \in \redsym$ can be described as a convex combination of symbols in $\spanningsym$ so that $\pmfcondx[\redidx] = \sum_{\spanningidx \in \spanningsym} \coeff \pmfcondx[\spanningidx]$ for some values $c_r$ s.t. $\sum_{r \in \centers} c_r = 1$. We have by the strict convexity of KL-divergence that
\begin{align*}
    \kl{\pmfcondx[\redidx]}{Q_\rvY^\star} &= \kl{\sum_{r\in \centers} c_r \pmfcondx[\redidx]}{Q_\rvY^\star} \\
    &\leq \sum_{r\in \centers} c_r \kl{\pmfcondx[\redidx]}{Q_\rvY^\star} \leq C
\end{align*}
Hence, the symbols in $\symrx$ did not contribute to the information radius of the channel, and can, hence, be removed without loss in capacity. Note that the capacity-achieving distribution still lies in the convex hull of $W_\centers$. This concludes the proof.
\end{proof}

\subsection{Proof of \cref{lemma:pseudo_capacity_diff}}

\begin{proof}

Using the definitions of the pseudo-simplex and pseudo-capacity, we can bound the difference of the pseudo-capacities of the channels as follows.
Let $\pseudocaod$ be the unique minimizer $\max_{r \in \centers} \kl{\pmfcondx[r]}{Q_\rvY}$ over the pseudo simplex $\mathcal{P}_\eta(\mathcal{Y})$, i.e., $\pseudocaod \define \argmin_{Q_\rvY \in \mathcal{P}_\eta(\mathcal{Y})} \max_{r \in \centers} \kl{\pmfcondx[r]}{Q_\rvY}$. Then, with $x^\star$ being any maximizer of $\kl{\pmfcondx[x]}{\pseudocaod}$ and $\tmprep$ being its representative, assuming $\pmfcondx[\tmprep]$ is at least supported where $\pmfcondx[x^\star]$ is, we have
\begin{align*}
    &C_\eta\left( \{\pmfcondx\}_{x\in \symtx} \right) - C_\eta\left( \{\pmfcondx[r]\}_{r\in \centers} \right) \\
    &\min_{Q_\rvY \in \mathcal{P}_\eta(\mathcal{Y})} \! \max_{x \in \symtx} \kl{\pmfcondx[x]}{Q_\rvY} \! - \!\!\!\!\!\!\! \min_{Q_\rvY \in \mathcal{P}_\eta(\mathcal{Y})} \! \max_{r \in \centers} \kl{\pmfcondx[r]}{Q_\rvY} \\
    &= \min_{Q_\rvY \in \mathcal{P}_\eta(\mathcal{Y})} \max_{x \in \symtx} \kl{\pmfcondx[x]}{Q_\rvY} - \max_{r \in \centers} \kl{\pmfcondx[r]}{\pseudocaod} \\
    &\leq \max_{x \in \symtx} \kl{\pmfcondx[x]}{\pseudocaod} - \max_{r \in \centers} \kl{\pmfcondx[r]}{\pseudocaod} \\
    &\overset{(aa)}{=} \kl{\pmfcondx[x^\star]}{\pseudocaod} - \kl{\pmfcondx[\tmprep]}{\pseudocaod} \\
    &= \begin{aligned}[t] &\sum_{y \in \symrx} \ycondx[x^\star] \log \frac{\ycondx[x^\star]}{\pseudocaod(y)} \\
    &- \ycondx[\tmprep] \log \frac{\ycondx[\tmprep]}{\pseudocaod(y)} \end{aligned} \\
    &\overset{(a)}{=} \begin{aligned}[t] &\sum_{y \in \symrx} \ycondx[x^\star] \log \frac{\ycondx[x^\star]}{\ycondx[\tmprep]} \\
    &+ \log \frac{\ycondx[\tmprep]}{\pseudocaod(y)} \left(\ycondx[x^\star] - \ycondx[\tmprep]\right) \end{aligned} \\
    &\overset{(b)}{=} \begin{aligned}[t] &\kl{\pmfcondx[x^\star]}{\pmfcondx[\tmprep]} \\
    &+ \sum_{y \in \symrx} \log \frac{\ycondx[\tmprep]}{\pseudocaod(y)} \left(\ycondx[x^\star] \! - \! \ycondx[\tmprep]\right) \end{aligned} \\
    &= \begin{aligned}[t] &\kl{\pmfcondx[x^\star]}{\pmfcondx[\tmprep]} \\
    &+ \sum_{y \in \symrx} \sqrt{\ycondx[\tmprep]} \log \frac{\ycondx[\tmprep]}{\pseudocaod(y)} \\
    &\cdot \frac{\left(\ycondx[x^\star] - \ycondx[\tmprep]\right)}{\sqrt{\ycondx[\tmprep]}} \end{aligned} \\
    &\overset{(c)}{\leq} \begin{aligned}[t] &\kl{\pmfcondx[x^\star]}{\pmfcondx[\tmprep]} \\
    &+ \sqrt{\sum_{y \in \symrx} \ycondx[\tmprep] \log^2 \frac{\ycondx[\tmprep]}{\pseudocaod(y)}} \\
    &\cdot \sqrt{\sum_{y \in \symrx} \frac{\left(\ycondx[x^\star] - \ycondx[\tmprep]\right)^2}{\ycondx[\tmprep]}} \end{aligned} \\
    &\overset{(d)}{\leq} \begin{aligned}[t] &\chisq{\pmfcondx[x^\star]}{\pmfcondx[\tmprep]} \\
    &+ \sqrt{\E_{\pmfcondx[\tmprep]} \left[\log^2 \frac{\ycondx[\tmprep]}{\pseudocaod(y)}\right]} \\
    &\cdot \sqrt{\chisq{\pmfcondx[x^\star]}{\pmfcondx[\tmprep]}} \end{aligned} %
\end{align*}
where $(aa)$ is because for every $r\in \centers$ including $\tmprep$ we have $\kl{\pmfcondx[r]}{\pseudocaod} < \max_{r \in \centers} \kl{\pmfcondx[r]}{\pseudocaod}$. Further $(a)$ holds since $ab - cd = ab - ad + ad - cd = a (b-d) + d (a-c)$, $(b)$ follows from rearranging the terms, and $(c)$ holds by Cauchy–Schwarz, and $(d)$ holds since $\chi^2$ divergence is an upper bound to KL-divergence and the definition of the $\chi^2$-divergence. It remains to bound the second moment of $\log \frac{\ycondx[\tmprep]}{\pseudocaod(y)}$.

Let $\pmfcondxpruned[\tmprep]$ be the probability distribution over the subset of symbols in $\symrx$ that correspond to non-zero entries in $\pmfcondx[\tmprep]$, using $\lbp \triangleq \min_{y \in \symrx: \ycondx[\tmprep] \neq 0} \left(\ycondx[\tmprep]\right)$ as the smallest non-zero number in $\pmfcondx[\tmprep]$, we have
\begin{align*}
    &\sum_{y \in \symrx} \ycondx[\tmprep] \log^2 \frac{\ycondx[\tmprep]}{\pseudocaod(y)} \\
    &= \E_{\pmfcondx[\tmprep]} \left[\log^2 \frac{\ycondx[\tmprep]}{\pseudocaod(y)} \right] \\
    &\overset{(e)}{\leq} \begin{aligned}[t] &\E_{\pmfcondx[\tmprep]} \left[\log \frac{\ycondx[\tmprep]}{\pseudocaod(y)} \right]^2 \\ &+ \var_{\pmfcondx[\tmprep]} \left[\log \frac{\ycondx[\tmprep]}{\pseudocaod(y)} \right] \end{aligned} \\
    &\overset{(f)}{=} \begin{aligned}[t] &\E_{\pmfcondx[\tmprep]} \left[\log \frac{\ycondx[\tmprep]}{\pseudocaod(y)} \right]^2 \\& + \var_{\pmfcondxpruned[\tmprep]} \left[\log \frac{\ycondx[\tmprep]}{\pseudocaod(y)} \right] \end{aligned} \\
    &\overset{(g)}{\leq} C(\prunedchannel)^2 + \left(\log(1/\eta)-C(\prunedchannel)\right) \left(C(\prunedchannel)-\log(\lbp)\right) \\
    &= C(\prunedchannel) \log\left(\frac{\lbp}{\eta}\right) + \log\left(\frac{1}{\eta} \right) \log\left(\frac{1}{\lbp} \right),
\end{align*}
where $(e)$ holds from the decomposition of second moments in terms of first moment and variance. $(f)$ holds by the derivations in the sequel, and $(g)$ holds by applying Bhatia–Davis inequality where the variance calculation is limited to non-zero entries in $\pmfcondx[\tmprep]$. Hence, the random variable $\log \frac{\ycondx[\tmprep]}{\pseudocaod(y)}$ is bounded as
\begin{align*}
    &\log(\lbp) \leq \min_{y \in \symrx: \ycondx[\tmprep] \neq 0} \log \frac{\ycondx[\tmprep]}{\pseudocaod(y)} \\
    &\leq \log \frac{\ycondx[\tmprep]}{\pseudocaod(y)} \\
    &\leq \max_{y \in \symrx: \ycondx[\tmprep] \neq 0} \log \frac{\ycondx[\tmprep]}{\pseudocaod(y)} \leq \log\left(\frac{1}{\eta}\right).
\end{align*}
$(f)$ follows from $(e)$ since
\begin{align*}
    &\var_{\pmfcondx[\tmprep]} \left[\log \frac{\ycondx[\tmprep]}{\pseudocaod(y)} \right] \\
    &= \begin{aligned}[t] &\E_{\pmfcondx[\tmprep]} \left[\log \frac{\ycondx[\tmprep]}{\pseudocaod(y)} \right]^2 \\
    &+ \E_{\pmfcondx[\tmprep]} \left[\log^2 \frac{\ycondx[\tmprep]}{\pseudocaod(y)} \right] \end{aligned} \\
    &= \begin{aligned}[t] &\sum_{y \in \symrx} \ycondx[\tmprep] \left[\log \frac{\ycondx[\tmprep]}{\pseudocaod(y)} \right]^2 \\
    &+ \sum_{y \in \symrx} \ycondx[\tmprep] \left[\log^2 \frac{\ycondx[\tmprep]}{\pseudocaod(y)} \right] \end{aligned} \\
    &\overset{(a)}{=} \begin{aligned}[t] &\sum_{y \in \symrx: \ycondx[\tmprep] \neq 0} \!\!\!\!\!\!\! \ycondx[\tmprep] \left[\log \frac{\ycondx[\tmprep]}{\pseudocaod(y)} \right]^2 \\
    &+ \sum_{y \in \symrx: \ycondx[\tmprep] \neq 0} \!\!\!\!\!\!\! \ycondx[\tmprep] \left[\log^2 \frac{\ycondx[\tmprep]}{\pseudocaod(y)} \right] \end{aligned} \\
    &\overset{(b)}{=} \begin{aligned}[t] &\E_{\pmfcondxpruned[\tmprep]} \left[\log \frac{\ycondx[\tmprep]}{\pseudocaod(y)} \right]^2 \\
    &+ \E_{\pmfcondxpruned[\tmprep]} \left[\log^2 \frac{\ycondx[\tmprep]}{\pseudocaod(y)} \right] \end{aligned} \\
    &= \var_{\pmfcondxpruned[\tmprep]} \left[\log \frac{\ycondx[\tmprep]}{\pseudocaod(y)} \right],
\end{align*}
where $(a)$ holds since $x\log x = 0$ and $(b)$ holds since $\pmfcondx[\tmprep]$ is a valid probability distribution over a subset of $\symrx$. 
This concludes the proof.
\end{proof}

\subsection{Proof of \cref{thm:closetohull}}

To prove \cref{thm:closetohull}, we rely on the following lemma to turn the conditional distributions in the convex hull that are close to symbols outside the convex hull in terms of $\chi^2$-divergence into their nearest neighbors.

\begin{lemma} \label{lemma:pseudo_capacity_diff_hull}
    Let $\centers$ be a set of symbols whose conditional distributions span the convex hull $\hull{\{\pmfcondx[r]\}_{r \in \centers}}$ such that each symbol not in the convex hull $\notinconvexhullidx \in \notinconvexhullsym \subset \symtx \setminus \centers$ has a bounded difference of $\hulldev{\centers}{\notinconvexhullidx}$ in terms of $\chi^2$ divergence. Let for each $\notinconvexhullidx \in \notinconvexhullsym$ the distribution $\closestconditional[\notinconvexhullidx]$ be the $\chi^2$-closest of $\pmfcondx[\notinconvexhullidx]$ to the convex hull $\hull{\{\pmfcondx[r]\}_{r \in \centers}}$, then we have    
    \begin{align}
        C&\left(\{\pmfcondx[x]\}_{x \in \symtx} \right) - C\left(\{\pmfcondx[r]\}_{r\in \centers} \right) \nonumber \\
        &\leq \begin{aligned}[t] &C\left(\{\pmfcondx[r]\}_{r\in \centers} \cup \{\pmfcondx[\notinconvexhullidx]\}_{\notinconvexhullidx \in \notinconvexhullsym}\right) \nonumber \\
        &- C(\{\pmfcondx[r]\}_{r\in \centers} \cup \{\closestconditional\}_{\notinconvexhullidx \in \notinconvexhullsym}). \nonumber \end{aligned}
\end{align}
\end{lemma}
\begin{proof}
    We prove this lemma in \cref{app:proof_pseudo_capacity_diff_hull}.
\end{proof}
By applying the above lemma, bounding the capacity loss based on the $\chi^2$ distance of each symbol outside the convex hull to its nearest neighbor in $\centers$ is equivalent to considering the distance to the convex hull, since we can artificially transform the channel to contain all the points that minimize the distance to the convex hull to nearest neighbors that are actually part of the channel.

\begin{proof}[Proof of \cref{thm:closetohull}]
Let $\pseudocaod$ be the unique minimizer of
\begin{align*}
    \min_{Q_\rvY \in \mathcal{P}_\eta(\mathcal{Y})} \max_{r \in \centers} \kl{\pmfcondx[r]}{Q_\rvY}.
\end{align*}

Then, using triangle inequality, \cref{lemma:pseudo_capacity_diff}, and results from \cite[Lemma 1]{egger2023maximal}, we can write for any $x^\star \in \symtx$ that maximizes $\kl{\pmfcondx[x]}{\pseudocaod(\prunedchannel)}$ and the distance to its nearest neighbor, which at the same time is the distribution $\closestconditional[x^\star]$ on the convex hull closest to $\pmfcondx[x^\star]$ in terms of $\chi^2$-distance ($\chisq{\pmfcondx[x^\star]}{\pmfcondx[{\tmprep[x^\star]}]}$, that %
\begin{align*}
    &\min_{Q_\rvY \in \mathcal{P}(\mathcal{Y})} \max_{x \in \symtx} \kl{\pmfcondx[x]}{Q_\rvY} - \max_{r \in \centers} \kl{\pmfcondx[r]}{Q_{Y}^\star} \\
    &\leq \!\!\! \begin{aligned}[t] &\bigg| \! \min_{Q_\rvY } \! \max_{x \in \symtx} \kl{\pmfcondx[x]}{Q_\rvY} \! - \!\!\!\!\!\!\!\! \min_{Q_\rvY \in \mathcal{P}_\eta(\mathcal{Y})} \! \max_{x \in \symtx} \kl{\pmfcondx[x]}{Q_\rvY} \!\! \bigg| \\
    &\!\!\!\!\!\!\!\! +\bigg| \! \min_{Q_\rvY \in \mathcal{P}_\eta(\mathcal{Y})} \!\max_{x \in \symtx} \kl{\pmfcondx[x]}{Q_\rvY} \! - \! \max_{r \in \centers} \kl{\pmfcondx[r]}{\pseudocaod} \!\! \bigg| \\
    &+ \bigg| \max_{r \in \centers} \kl{\pmfcondx[r]}{\pseudocaod} - \max_{r \in \centers} \kl{\pmfcondx[r]}{Q_\rvY^\star} \bigg| \end{aligned} \\
    &\overset{(a)}{\leq} \!\! 4\card{\mathcal{Y}} \eta + \!\!\!\!\!\!\! \min_{Q_\rvY \in \mathcal{P}_\eta(\mathcal{Y})} \max_{x \in \symtx} \kl{\pmfcondx[x]}{Q_\rvY} \! \\
        & \qquad\qquad\qquad -\max_{r \in \centers} \kl{\pmfcondx[r]}{\pseudocaod} \\
    &\overset{(b)}{\leq} \begin{aligned}[t] &4\card{\mathcal{Y}} \eta + \chisq{\pmfcondx[x^\star]}{\pmfcondx[\tmprep]} \\
    &+ \sqrt{C_\eta(\prunedchannel) \log\left(\frac{\lbp}{\eta}\right) + \log\left(\frac{1}{\eta} \right) \log\left(\frac{1}{\lbp} \right)} \\
    &\cdot \sqrt{\chisq{\pmfcondx[x^\star]}{\pmfcondx[\tmprep]}} \end{aligned} \\
    &\overset{(c)}{\leq} \begin{aligned}[t] &4\card{\mathcal{Y}} \eta + \chisq{\pmfcondx[x^\star]}{\pmfcondx[\tmprep]} \\
    &+ \sqrt{\log\left(\frac{1}{\eta}\right) \left( C(\prunedchannel) + \log\left(\frac{1}{\lbp} \right)\right)} \\
    &\cdot \sqrt{\chisq{\pmfcondx[x^\star]}{\pmfcondx[\tmprep]}}, \end{aligned}
\end{align*}
where $(a)$ follows from applying \cite[Lemma 1]{egger2024maximal} twice; $(b)$ follows from \cref{lemma:pseudo_capacity_diff}; and $(c)$ is by bounding $C_\eta(\prunedchannel) \leq C(\prunedchannel)$ and from $\log\left(\frac{\kappa}{\eta}\right) \leq \log\left(\frac{1}{\eta}\right)$.

This emits a bias-variance-type trade-off based on the parameter $\eta$, and it remains to appropriately choose $\eta$. The function 
\begin{align*}
    \sqrt{\log\left(\frac{1}{\eta}\right) \left( C(\prunedchannel) + \log\left(\frac{1}{\lbp} \right)\right)}
\end{align*}
is strictly decreasing in $\eta$, with a sharp decrease around $0$. With constants $\tradeconsta \define 4\card{\symrx}$ and $\tradeconstb \define \sqrt{C(\prunedchannel) - \log \left(\lbp \right)}$, we find a good $\eta$ by solving
\begin{align*}
    \eta^\star &= \argmin_{\eta} \, \tradeconsta \eta + \tradeconstb \sqrt{\log\left(\frac{1}{\eta}\right)}.
\end{align*}
Therefore, we analyze the derivatives w.r.t. $\eta$:
\begin{align*}
    \frac{\partial \sqrt{\log\left(\frac{1}{\eta}\right)}}{\partial \eta} %
    = -\frac{1}{2\eta \sqrt{\log\left(\frac{1}{\eta}\right)}} \approx -\frac{1}{\sqrt{\eta} - 0.07},
\end{align*}
where the approximation is valid in the regime of interest ($\eta \ll 1$). Hence, we can choose $\eta$ to equalize the derivatives of the above and the linear dependency $4\eta\card{\symrx}$ on $\eta$. We have
\begin{align*}
    &\tradeconsta = \tradeconstb \cdot \sqrt{\chisq{\pmfcondx[x^\star]}{\pmfcondx[\tmprep]}} \cdot \frac{1}{\sqrt{\eta}-0.07} \Leftrightarrow \\
    & \eta = \left(\frac{\tradeconstb \cdot \sqrt{\chisq{\pmfcondx[x^\star]}{\pmfcondx[\tmprep]}}}{\tradeconsta} + 0.07\right)^2 %
\end{align*}

However, the value for $\eta$ is conditioned on the choice of the symbol $x^\star$, which is supposed to be the symbol that maximizes $\kl{\pmfcondx[x]}{\pseudocaod(\prunedchannel)}$. To determine this symbol, we must fix $\eta$ on the other hand. Therefore, we use the fact that the bound holds uniformly for all $\eta$. Hence, we choose for the calculation of $\eta$ the symbol that maximizes the difference to the capacity achieving output distribution $\kl{\pmfcondx[x]}{Q_\rvY^\star(\prunedchannel)}$, which is expected to be close to $\kl{\pmfcondx[x]}{\pseudocaod(\prunedchannel)}$ and consequently a good choice for computing $\eta$.
This concludes the proof of the theorem.
\end{proof}

\subsection{Proof of \cref{lemma:pseudo_capacity_diff_hull}} \label{app:proof_pseudo_capacity_diff_hull}

\begin{proof}
    Let $\inconvexhullsym \subset \symtx \setminus \centers$ be the symbols whose conditionals $\pmfcondx[\inconvexhullidx], \inconvexhullidx \in \inconvexhullsym$ are contained in the convex hull of the conditionals of symbols in $\centers$. Let for each symbol $\notinconvexhullidx \in \notinconvexhullsym \subset \symtx \setminus \centers$ whose conditional $\pmfcondx[\inconvexhullidx]$ is not contained in the convex hull of symbols in $\centers$ the distribution $\closestconditional$ in the convex hull be closest to $\pmfcondx[\notinconvexhullidx]$ in terms of $\chi^2$ distance, i.e.,
    \begin{align*}
        \closestconditional = \argmin_{\closestconditionalvar \in \hull{\{\pmfcondx[r]\}_{r \in \centers}}} \chisq{\pmfcondx[\notinconvexhullidx]}{\closestconditionalvar}.
    \end{align*}
    Since symbols are either in the convex hull or not, all sets are disjoint and we have that $\symtx = \centers \cup \inconvexhullsym \cup \notinconvexhullsym$. Consider a set of conditional probabilities containing those given by symbols in $\centers$ and $\inconvexhullsym$ and those given by distributions in the convex hull of $\centers$ closest to all symbols in $\notinconvexhullsym$. Then we can bound the capacity difference as follows:
    \begin{align}
        & C\left(\{\pmfcondx[x]\}_{x \in \symtx} \right) - C\left(\{\pmfcondx[r]\}_{r\in \centers} \right) \nonumber \\
        &= C\left(\{\pmfcondx[r]\}_{r\in \centers} \! \cup \! \{\pmfcondx[\notinconvexhullidx]\}_{\notinconvexhullidx \in \notinconvexhullsym} \! \cup \! \{\pmfcondx[\inconvexhullidx]\}_{\inconvexhullidx \in \inconvexhullsym}\right) \label{eqline:ineq1a} \\
        &- C\left(\{\pmfcondx[r]\}_{r\in \centers} \cup \{\pmfcondx[\notinconvexhullidx]\}_{\notinconvexhullidx \in \notinconvexhullsym} \right) \label{eqline:ineq1} \\
        &+ C\left(\{\pmfcondx[r]\}_{r\in \centers} \cup \{\pmfcondx[\notinconvexhullidx]\}_{\notinconvexhullidx \in \notinconvexhullsym} \right) \nonumber %
        \\
        &- C(\{\pmfcondx[r]\}_{r\in \centers} \cup \{\closestconditional\}_{\notinconvexhullidx \in \notinconvexhullsym}) \nonumber \\
        &+ \! C(\{\pmfcondx[r]\}_{r\in \centers} \! \cup \! \{\closestconditional\}_{\notinconvexhullidx \in \notinconvexhullsym}) \! - \! C\left(\!\{\pmfcondx[r]\}_{r\in \centers}\! \right) \label{eqline:ineq4} \\
        &\leq C\left(\{\pmfcondx[r]\}_{r\in \centers} \cup \{\pmfcondx[\notinconvexhullidx]\}_{\notinconvexhullidx \in \notinconvexhullsym}\right) \nonumber \\
        &- C(\{\pmfcondx[r]\}_{r\in \centers} \cup \{\closestconditional\}_{\notinconvexhullidx \in \notinconvexhullsym}), \nonumber
    \end{align}
    where \eqref{eqline:ineq1a}+\eqref{eqline:ineq1} and \eqref{eqline:ineq4} are $0$ by the application of \cref{thm:inconvexhull}. This proves the equivalence of comparing the distance of removed symbols to either the nearest neighbor, or the closest distribution in the convex hull. %
\end{proof}

\end{document}